\documentclass[conference,letterpaper]{IEEEtran}

\usepackage[utf8]{inputenc} 
\usepackage[T1]{fontenc}
\usepackage{url}
\usepackage{ifthen}
\usepackage{cite}
\usepackage[cmex10]{amsmath} 

\usepackage{amsthm}
\usepackage{amssymb}
\usepackage{amsmath}
\usepackage{algorithmic}
\usepackage{multicol}
\usepackage[colorlinks=true, pdfstartview=FitV, linkcolor=blue,
citecolor=blue, urlcolor=blue]{hyperref}
\newtheorem{theorem}{Theorem}[section]
\newtheorem{corollary}{Corollary}[theorem]
\newtheorem{remark}{Remark}
\newtheorem{lemma}[theorem]{Lemma}
\newtheorem{definition}{Definition}[section]
\usepackage{hyperref}
\usepackage{caption}
\newcommand{\subparagraph}{}
\usepackage{titlesec}
\usepackage{cleveref}
\usepackage[ruled,vlined]{algorithm2e}

\newtheorem*{claim}{Claim}

\usepackage{tikz}
\usepackage[framemethod=tikz]{mdframed}
\usepackage{graphicx}
\usepackage{array}
\newcolumntype{P}[1]{>{\centering\arraybackslash}p{#1}}
\newcommand{\E}{\ensuremath{\operatorname{\mathbb{E}}}}
\newcommand{\real}{\ensuremath{\operatorname{\mathbb{R}}}}
\newcommand{\mprob}{\ensuremath{\operatorname{\mathbb{P}}}}

\begin{document}
\title{Differentially Private Synthetic Data with Private Density Estimation} 


\author{%
  \IEEEauthorblockN{Nikolija Bojkovic}
  \IEEEauthorblockA{School of Electrical Engineering \\
                   University of Belgrade \\
                   Belgrade, Serbia \\
                   Email: nikolijabojkovic98@gmail.com}
  \and
  \IEEEauthorblockN{Po-Ling Loh}
  \IEEEauthorblockA{Department of Pure Mathematics and Mathematical Statistics \\ 
                    University of Cambridge \\
                    Cambridge, United Kingdom \\
                    Email: pll28@cam.ac.uk}
}

\maketitle

\begin{abstract}
The need to analyze sensitive data, such as medical records or financial data, has created a critical research challenge in recent years. In this paper, we adopt the framework of differential privacy, and explore mechanisms for generating an entire dataset which accurately captures characteristics of the original data.
We build upon the work of Boedihardjo et al.~\cite{BoeEtal22}, which laid the foundations for a new optimization-based algorithm for generating private synthetic data.
Importantly, we adapt their algorithm by replacing a uniform sampling step with a private distribution estimator; this allows us to obtain better computational guarantees for discrete distributions, and develop a novel algorithm suitable for continuous distributions.
We also explore applications of our work to several statistical tasks.
\end{abstract}

\section{Introduction}

We live in an era where it is increasingly important to preserve data privacy, considering the potential harm associated with unauthorized sharing of personal information. Nonetheless, instances arise where accessing and utilizing data is imperative for purposes such as medical research relying on medical records, or social network analysis requiring network data. However, traditional anonymization methods may fall short, leading to the need for more robust solutions.
The concept of differential privacy was first introduced by Dwork~\cite{Dwo06}, with subsequent papers~\cite{DwoEtal06,NikEtal13,NisEtal07, LiEtal10, MohEtal11, DwoEtal14,Blu13} developing important new theoretical notions regarding differential privacy. Significant work has been done on the private answering of dataset queries using differential privacy~\cite{BarEtal07, HarRot10}.
However, combining differential privacy with synthetic data release enables a broader, more flexible analysis, since downstream tasks such as clustering, classification, regression, and visualization may be conducted while still preserving privacy, allowing exploration beyond a predetermined set of queries. 
This has far-reaching implications for data-driven research and decision-making across fields such as healthcare, finance, and social science~\cite{JorEtal22}.  
 
 




The starting point of our work is an algorithm for generating differentially private synthetic data proposed by Boedihardjo et al.~\cite{BoeEtal22}, for which detailed steps are provided in Section~\ref{SecBackground}. Boedihardjo et al.~\cite{BoeEtal22} provided theoretical guarantees for the accuracy of their private dataset (measured in terms of a uniform bound on the errors of empirical averages of test functions) when the data lie in a Boolean space ($\Omega=\{0,1\}^p$). This allows them to take the reference distribution $\mu$ used to generate the initial data in their algorithm as uniform. The main contribution of our work is to analyze alternative choices for $\mu$, i.e., an initial estimate of the data-generating distribution, and consequences for private synthetic data generation over more general spaces. In Section~\ref{SecDiscrete}, we study discrete datasets defined on $\Omega= \{0,1,2,\ldots,t-1\}^p$, using a private histogram estimator~\cite{WasZho10} for $\mu$; in Section~\ref{SecContinuous}, we study continuous datasets defined on $\Omega = [0,1]^p$, using a private kernel density estimator~\cite{AldRub17} for $\mu$. Our methods provide rigorous guarantees for privacy, accuracy, and computational efficiency. Notably, we are able to demonstrate concrete computational gains in comparison to the original proposal with uniform sampling~\cite{BoeEtal22} in the discrete setting---whereas in the continuous setting, we propose an algorithm which fills an important theoretical gap where the uniform sampling analysis would not apply.


Other existing approaches for private synthetic data generation with theoretical guarantees are relatively scarce. We highlight the follow-up papers by the same authors~\cite{BoeEtal22b, BoeEtal22c} and an approach based on kernel mean embeddings~\cite{BalEtal18}, and also note that earlier work on private synthetic data involved algorithms which were not computationally tractable~\cite{HarRot10, HarEtal12} due to their focus on worst-case rather than average-case datasets~\cite{UllVad11}. A separate line of work leverages generative adversarial networks~\cite{LuEtal17,JorEtal18,XuEtal19, BeaEtal19, DocEtal23}. However, we note that although these papers derive rigorous privacy guarantees on the data generated by the algorithms, the accuracy of the synthetic data is only evaluated empirically. The paper \cite{McKEtal21} also only contains empirical comparisons of accuracy, but the algorithm proposed there bears some resemblance to the algorithms studied in our paper: In particular, they solve an optimization problem which finds the best distribution in a class that matches the marginals of the observed data, subject to noisy  perturbations. However, the optimization is restricted to a set of distributions characterized by a class of graphical models. Finally, we mention the paper \cite{LinEtal23}, which recently proposed an interesting method based on ``private evolution" that also bears some similarity to our approach. Their algorithm involves iteratively generating samples from a (private) learned histogram, which is iteratively reweighted in such a way that it eventually approaches the histogram of the private data. The histogram is then used to generate data that are fed into an application programming interface (API) for further inferential tasks. The authors provide a theoretical bound on the accuracy of the private histogram in terms of Wasserstein distance. We note that while the idea of reweighting randomly generated data in a private manner to estimate a histogram of the observed data is common both in \cite{LinEtal23} and our work, our algorithms do not involve iterative fitting and seek to quantify accuracy in a different sense.

\textbf{Notation:} For a matrix $A \in \real^{p \times p}$, we write $\|A\|_2$ to denote the spectral norm. For $x \in \real$, we write $(x)_+ := \max\{x, 0\}$. For two distributions $p$ and $q$, we write $\kappa(p\|q)$ to denote the R\'{e}nyi condition number, and write $TV(p,q)$ to denote the total variation distance. We use the notation $B(x,r) := \{x' \in [0,1]^p \mid \|x - x'\|_2 \leq r, r > 0\}$, and let $v_p$ denote the volume of the unit ball $B(0,1)$.


\section{Background}
\label{SecBackground}


We first review the basic definition of differential privacy, and then provide a definition of how we will quantify accuracy of a synthetic dataset.

\begin{definition}[Privacy]
A randomized function $M$ is $\epsilon$-differentially private if, for all datasets $X$ and $X'$ differing in one element and all measurable sets $S \subseteq \text{range}(M)$, we have $P[M(X) \in S] \leq e^{\epsilon} \cdot P[M(X') \in S]$.
 \end{definition}
The following notion of accuracy is proposed in \cite{BoeEtal22}:
\begin{definition}[Accuracy]
\label{DefAccuracy}
Consider a dataset $X = \{x_i\}_{i=1}^n$ and a finite class $\mathcal{F}$ of functions from $\Omega$ to $[-1, 1]$.
A synthetic dataset $Y = \{y_j\}_{j=1}^k$ is $\delta$-accurate if
\[
\max_{f \in \mathcal{F}}\left| \frac{1}{k} \sum_{j=1}^{k} f(y_j) - \frac{1}{n} \sum_{i=1}^{n} f(x_i) \right| \leq \delta. 
\]
\end{definition}
Our technical results will require $\mathcal{F}$ to be a class of bounded functions from $\Omega$ to $[-1,1]$, which we assume for the remainder of the paper.

A particular case of interest is where $\mathcal{F}$ is the set containing all functions that encode marginals up to degree $d$, in which case $|\mathcal{F}| = \binom{p}{\leq d} := \binom{p}{0} + \binom{p}{1} + \cdots + \binom{p}{d}$. The idea of looking at marginals is common \cite{ThaEtal12,BoeEtal22b,BarEtal07}, as it ensures that no individual's data significantly influences the final output of the algorithm. Note that other methods to evaluate the accuracy of private algorithms include $L_1$-sensitivity~\cite{DwoEtal06, MohEtal11}, query matrix sensitivity~\cite{NisEtal07}, and computing frequent items~\cite{BarEtal21}.

Algorithm~\ref{AlgBoeEtal}, proposed by Boedihardjo et al.~\cite{BoeEtal22}, is the starting point of our work. Briefly, the algorithm proceeds by sampling data from a reference distribution $\mu$ (which is taken independent of the data) to create a reduced space $\Omega^* \subseteq \Omega^m$. A histogram $h^*$ is then obtained over $\Omega^*$ by solving an optimization problem to obtain a distribution which best fits noisy versions of empirical averages of functions $f \in \mathcal{F}$. Synthetic data are generated by sampling from $h^*$. The algorithm has the following theoretical guarantee, where we recall the definition $\kappa(\nu\parallel\mu) = \int \left(\frac{d\nu}{d\mu}\right)^2 d\mu$ of the R\'{e}nyi condition number:


\begin{algorithm}
    \caption{Private synthetic data generation~\cite{BoeEtal22}}\label{AlgBoeEtal}
    \begin{algorithmic}
        \STATE \textbf{Input:} Dataset $X=\left(x_1, \ldots, x_n\right) \subseteq \Omega^n$, family $\mathcal{F}$\\ of test functions from $\Omega$ to $[-1,1]$,
        parameter $\sigma>0$  
        \STATE 1. Sample i.i.d.\ data $\Omega^*= \{z_1, \dots, z_m\} \subseteq \Omega^m$ with probability measure $\mu$
                \STATE 2. For each test function $f \in \mathcal{F}$, \\
generate an independent Laplacian random variable $\lambda(f) \sim \operatorname{Lap}(\sigma)$
                \STATE 3. Compute a density $h^*$ by solving
\begin{multline*}
h^*=\operatorname{argmin}\Bigg\{\max _{f \in \mathcal{F}}\Bigg| \sum_{i=1}^m f\left(z_i\right) h\left(z_i\right) \\
-\frac{1}{n} \sum_{i=1}^n f\left(x_i\right)-\lambda(f)\Bigg|: h \text { is a density on } \Omega^*\Bigg\} .
\end{multline*}
                \STATE 4. Draw i.i.d.\ data $Y=\left(y_1, \ldots, y_k\right)$\\ according to $h^*$
                \STATE \textbf{Output}: Synthetic data $Y=\left(y_1, \ldots, y_k\right)$
    \end{algorithmic}
\end{algorithm}
%
\begin{theorem}[Theorems 2.2 \& 2.3 of Boedihardjo et al.~\cite{BoeEtal22}]
\label{ThmBoeEtal}
Let $\delta, \gamma > 0$ and set $\sigma = \frac{\delta}{\log(|\mathcal{\mathcal{F}}|/\gamma)}$.
\begin{enumerate}
\item If $n \geq \frac{2|\mathcal{F}|}{\epsilon \delta} \log\left(\frac{|\mathcal{F}|}{\gamma}\right)$, Algorithm~\ref{AlgBoeEtal} is $\epsilon$-differentially private.
\item Let $\min\{n, k\} \geq \frac{1}{\delta^2} \log\left(\frac{|\mathcal{F}|}{\gamma}\right)$ and $m \geq \frac{K|\mathcal{F}|}{\delta^2 \gamma}$, where $\delta \in (0, 1/2]$ and $\gamma \in (0, 1/4)$. Suppose $X$ is sampled i.i.d.\ from $\Omega$ according to a probability measure $\nu$.
Assume $\mu$ satisfies $\kappa(\nu\|\mu) \leq K$. Then with probability at least $1 - 4\gamma$, the synthetic dataset $Y$ generated by Algorithm~\ref{AlgBoeEtal} is $8\delta$-accurate.
\end{enumerate}
\end{theorem}




When the data are Boolean ($\Omega=\{0,1\}^p$), Boedihardjo et al.~\cite{BoeEtal22} suggest taking $\mu$ to be a uniform measure, in which case we can choose $K = |\Omega|$ in part (2) of Theorem~\ref{ThmBoeEtal}. See Theorem~\ref{ThmBoeEtalUni} in Appendix~\ref{AppUniform} for more details.
\begin{remark}

The guarantees of Theorem~\ref{ThmBoeEtal} shed some light on the benefits of private synthetic data. Although Algorithm~\ref{AlgBoeEtal} bears resemblance to the Laplace mechanism, in the sense that independent Laplacian noise $\lambda(f)$ is added for each $f \in \mathcal{F}$, the overall \emph{storage} cost of synthetic data is only $k = O(\log |\mathcal{F}|)$, as opposed to $O(|\mathcal{F}|)$, which would be the cost of storing the answers to all individual queries. This is of course only assured if the sample size is large enough, as stipulated in part (1) of the theorem. Furthermore, there are additional computational costs of solving the optimization program in step (2) of the algorithm, but as we discuss in more detail in Section~\ref{SecComp}, the program can be recast as a linear program and solved efficiently in practice.
\end{remark}

On the other hand, choosing $\mu$ to be uniform may not be optimal. Furthermore, in the continuous case, we cannot upper-bound the R\'{e}nyi condition number with the same techniques. Table~\ref{TableSummary} compares the private synthetic data generation methods studied in our paper.

\begin{table*}[!ht]
    \centering
    \resizebox{\textwidth}{!}{%
    \begin{tabular}{|p{0.21\linewidth}|p{0.11\linewidth}|p{0.09\linewidth}|p{0.26\linewidth}|p{0.07\linewidth}|p{0.12\linewidth}|}
    \hline
         \multicolumn{1}{|c|}{\textbf{Mechanism}}  & \multicolumn{1}{c|}{\textbf{$\Omega$}} & \multicolumn{1}{c|}{\textbf{$|X|$}} & \multicolumn{1}{c|}{\textbf{Description}} & \multicolumn{1}{c|}{\textbf{$|Y|$}} & \multicolumn{1}{c|}{\textbf{Linear Program}} \\ \hline
        Uniform sampling (Theorem~\ref{ThmBoeEtalUni}) &
        $\{0,1,...,t-1\}^p$ & $O(|\mathcal{F}|\log |\mathcal{F}|)$  &
        Algorithm~\ref{AlgBoeEtal} is run with $\mu$ uniform & $O(\log |\mathcal{F}|)$  &
        $|\mathcal{F}|+O(|\Omega|)+1$\\ \hline
         Sampling from a private histogram (Theorem~\ref{t1}) & $\{0,1,...,t-1\}^p$ & $O(|\mathcal{F}|\log |\mathcal{F}|)$ & Algorithm~\ref{AlgBoeEtal} is run with $\mu$ privately estimated from the empirical measure $\nu$ using the perturbed histogram method &  $O(\log |\mathcal{F}|)$ & $|\mathcal{F}|+O(1)+1$ \\ \hline
        Sampling using kernel density estimation (Theorem~\ref{t2}) &  $[0,1]^p$ & $O(|\mathcal{F}|\log |\mathcal{F}|)$  & Algorithm~\ref{AlgBoeEtal} is run with $\mu$ privately estimated from the empirical measure $\nu$ using kernel density estimation & $O(\log |\mathcal{F}|)$  & $|\mathcal{F}|+O(1)+1$ \\ \hline
    \end{tabular}%
    }
\caption{Summary of main results}
\label{TableSummary}
\end{table*}
\section{Sampling from a private histogram}
\label{SecDiscrete}

Consider a discrete distribution on the set $\Omega = \{0, 1, 2, \ldots, t-1\}^p$. This extends the framework of Boedihardjo et al.~\cite{BoeEtal22} beyond Boolean data.
Rather than using a uniform distribution for $\mu$, as suggested by Boedihardjo et al.~\cite{BoeEtal22} (cf.\ Appendix~\ref{AppUniform}),
our idea is to
choose $\mu$ using the perturbed histogram estimator method from Wasserman and Zhou~\cite{WasZho10} (the detailed algorithm is included as Algorithm~\ref{AlgWasZho} in Appendix \ref{app:4}). We arrive at the following theorem:

\begin{theorem}\label{t1}
Let $\epsilon > 0$, $\delta \in (0, \frac{1}{2}]$, and $\gamma\in (0, \frac{1}{4})$. 
Suppose $X$ is sampled i.i.d.\ from $\Omega= \{0,1,2,\ldots,t-1\}^p$ according to a probability measure $\nu$. Define $\tau_1 := \min_x \nu(x)$ and $\tau_2 := \max_x \nu(x)$, and suppose $\tau_1 > 0$. Set $\sigma = \delta/\log\left(|\mathcal{F}|/\gamma\right)$,
Suppose $\mu$ is obtained using the perturbed histogram method with privacy parameter $\epsilon$. Suppose
\begin{align*}
    & n \geq \max \Bigg\{\frac{2|\mathcal{F}|}{\epsilon \delta} \log\left(\frac{|\mathcal{F}|}{\gamma}\right), \quad \frac{8}{\tau_1^2}\left(|\Omega| + 2\log \frac{1}{\gamma}\right), \notag \\
    & \qquad \qquad \qquad \frac{8\delta |\Omega| (1+\tau_2 |\Omega|)  \log |\Omega|}{\tau_1 \gamma \log(|\mathcal{F}|/\delta)} \Bigg\}, \nonumber \\
    & \min\{n, k\} \geq \frac{1}{\delta^2} \log\left(\frac{|\mathcal{F}|}{\gamma}\right), \quad m \geq \frac{2K|\mathcal{F}|}{\delta^2 \gamma}, \nonumber \\
    & \text{where }\nonumber \\
    & K = \left(1+C_1\left(\sqrt{\frac{2\log \frac{1}{\gamma}}{n}}+\sqrt{\frac{|\Omega|}{n}}\right)\right)^{\frac{1}{2}} \nonumber \\
    & \qquad \qquad \cdot \left(1+ \frac{C_2\delta |\Omega| (1+\tau_2 |\Omega|) \log |\Omega|}{\gamma n \log (|\mathcal{F}|/\delta)}\right)^{1/2}.
\end{align*}
Then Algorithm~\ref{AlgBoeEtal} outputs a $2\epsilon$-differentially private dataset $Y$ which, with probability at least $1 - 3\gamma - \exp(-cn)$, is $8\delta$-accurate.
\end{theorem}

The proof is contained in Appendix~\ref{AppT1}. Note that the privacy parameter becomes $2\epsilon$ due to basic composition~\cite{DwoEtal14} of privately estimating $\mu$, followed by applying Algorithm~\ref{AlgBoeEtal}.

\begin{remark}
Theorem~\ref{t1} can be applied to Boolean datasets if we take $t=2$. Compared to the uniform sampling strategy proposed in Boedihardjo et al.~\cite{BoeEtal22} (see Theorem~\ref{ThmBoeEtalUni}), our requirement on $m$ is much smaller, since $K$ is smaller ($O(1)$ rather than $O(|\Omega|)$), leading to faster computation.
\end{remark}


It is instructive to understand how the quantities appearing in the conditions of Theorem~\ref{t1} scale with, e.g., the size of $|\Omega|$. Hence, we state the following simple corollary in the case where the underlying probability distribution $\nu$ is uniform, in which case we can take $\tau_1 = \tau_2 = \frac{1}{|\Omega|}$:

\begin{corollary}
\label{cor}
Suppose $\nu$ is uniform and
\begin{align}
    & n \geq \max \Bigg\{\frac{2|\mathcal{F}|}{\epsilon \delta} \log\left(\frac{|\mathcal{F}|}{\gamma}\right), \quad 8|\Omega|^3 + 16 |\Omega|^2 \log \frac{1}{\gamma}, \notag \\
    & \qquad \qquad \qquad \frac{16 \delta |\Omega|^2 \log |\Omega|}{\gamma \log(|\mathcal{F}|/\delta)} \Bigg\}, \nonumber \\
    & \min\{n, k\} \geq \frac{1}{\delta^2} \log\left(\frac{|\mathcal{F}|}{\gamma}\right), \quad m \geq \frac{2K|\mathcal{F}|}{\delta^2 \gamma}, \nonumber \\
    & \text{where} \nonumber \\
    & K = \left(1+C_1\left(\sqrt{\frac{2\log \frac{1}{\gamma}}{n}}+\sqrt{\frac{|\Omega|}{n}}\right)\right)^{\frac{1}{2}} \notag \\
    & \qquad \qquad \cdot \left(1+\frac{2C_2 \delta |\Omega| \log |\Omega|}{\gamma n \log (|\mathcal{F}|/\delta)} \right)^{\frac{1}{2}}.
\end{align}
Then Algorithm~\ref{AlgBoeEtal} outputs a $2\epsilon$-differentially private dataset $Y$ which, with probability at least $1 - 3\gamma - \exp(-cn)$, is $8\delta$-accurate.
\end{corollary}

\begin{remark}
One might wonder why, if $\mu$ is already a privately estimated distribution with a guaranteed bound on its deviation from $\nu$, we cannot simply sample synthetic data from $\mu$ rather than running Algorithm~\ref{AlgBoeEtal}. Indeed, Wasserman and Zhou~\cite{WasZho10} propose to generate private synthetic data from $\mu$. If we recall our definition of accuracy (Definition~\ref{DefAccuracy}), we see that the deviation bounds between $\mu$ and $\nu$ are in terms of standard distributional distances such as TV and sup-norm. On the other hand, if our goal is to tightly control the supremum of empirical averages over a class $\mathcal{F}$, it may be advantageous to further solve an optimization problem, as stipulated in Algorithm~\ref{AlgBoeEtal}.
\end{remark}


\section{Sampling on a continuous domain using kernel density estimation}
\label{SecContinuous}

We now consider the case of a continuous domain, where $X$ is sampled from $[0,1]^p$. In order to choose $\mu$, we employ a private kernel density estimation method, proposed by Alda and Rubinstein~\cite{AldRub17}. For a dataset $X = (x_1, x_2, \ldots, x_n)$, consider the kernel density estimator
\begin{equation}
\label{EqnKDE}
\nu_1(x)=\frac{1}{nh^p}\sum_{i=1}^{n}K\left(\frac{H_0^{-\frac{1}{2}}(x-x_i)}{h}\right),
\end{equation}
where $K(u) = k(|u|)$, with $k$ a standard Gaussian kernel and $H = h^2H_0$, where $H_0 \in \real^{p \times p}$ and $h > 0$ is a bandwidth. We then apply the Bernstein mechanism (see Algorithm~\ref{AlgAldRub} in Appendix \ref{appcon}) with $F(\{x_i\}_{i=1}^n, x) = \nu_1(x)$. This leads to the following guarantees:

\begin{theorem}\label{t2}
Let $\epsilon > 0$, $\delta \in \left(0, \frac{1}{2}\right]$, and $\gamma \in \left(0, \frac{1}{4}\right)$. Suppose $X$ is sampled i.i.d.\ from $[0, 1]^p$
according to an $L$-Lipschitz probability measure $\nu$. Let $\sigma = \delta/ \log\left(|\mathcal{F}|/\gamma\right)$.
Suppose $\mu$ is obtained through the Bernstein mechanism with privacy parameter $\epsilon$, where we make the bandwidth choices $h=C^{'}\Big(\frac{\log n}{n}\Big)^\frac{1}{p}$ and $H_0 = I_p$. Suppose $\tau := \inf\limits_{x}\nu(x) > 0$ and
\begin{align}
    & n \geq \max \Bigg\{\frac{2|\mathcal{F}|}{\epsilon \delta} \log\left(\frac{|\mathcal{F}|}{\gamma}\right), \quad \left(\frac{3}{\tau} \right)^{p+1}\left(\frac{2}{\epsilon \sqrt{(2\pi)^p}} \log \frac{1}{\gamma}\right)\Bigg\}, \nonumber \\
    & \min\{n, k\} \geq \frac{1}{\delta^2} \log\left(\frac{|\mathcal{F}|}{\gamma}\right),  \quad m \geq \frac{2(1+ C_1' \tau)|\mathcal{F}|}{\delta^2 \gamma}. \nonumber
\end{align}
Then Algorithm~\ref{AlgBoeEtal} outputs a $2\epsilon$-differentially private dataset $Y$ which, with probability at least $1-\gamma-\frac{1}{n}$, is $8\delta$-accurate.
\end{theorem}

The proof is contained in Appendix~\ref{AppT2}.



Again, it is instructive to study the expressions appearing in Theorem~\ref{t2} in special cases. The following corollary specializes to the case when $\nu$ is uniform:

\begin{corollary}
Suppose $\nu$ is uniform and
     \begin{align}
    & n \geq \max \Bigg\{\frac{2|\mathcal{F}|}{\epsilon \delta} \log\left(\frac{|\mathcal{F}|}{\gamma}\right), 3^{p+1}\left(\frac{2}{\epsilon \sqrt{(2\pi)^p}} \log \frac{1}{\gamma}\right)\Bigg\}, \nonumber \\
    & \min\{n, k\} \geq \frac{1}{\delta^2} \log\left(\frac{|\mathcal{F}|}{\gamma}\right), \quad m \geq \frac{C_0|\mathcal{F}|}{\delta^2\gamma}.
\end{align}
Then Algorithm~\ref{AlgBoeEtal} outputs a $2\epsilon$-differentially private dataset $Y$ which, with probability at least $1-\gamma-\frac{1}{n}$, is $8\delta$-accurate.
\end{corollary}


\section{Applications}
\label{SecApplications}

We now describe several concrete examples of the theory we have developed in this paper. We also provide remarks regarding the practical implementation of the synthetic data generation mechanisms.

\subsection{Statistical tasks}

\paragraph{\textbf{Estimating binomial proportions.}} The first example we discuss is to estimate binomial proportions. The functions in $\mathcal{F}$ all have range $\{0,1\}$, and the theory developed in Section~\ref{SecDiscrete} (indeed, even in the original paper~\cite{BoeEtal22}) clearly applies. Suppose we wish to derive a confidence interval for some binomial proportion $\E[f(x_i)] = p^*$. An asymptotically valid confidence interval based on the original data would take the form $\frac{1}{n} \sum_{i=1}^n f(x_i) \pm O\left(\frac{1}{\sqrt{n}}\right)$. If we were instead to compute a confidence interval for $p^*$ based on the synthetic dataset, the accuracy guarantee implies that $\frac{1}{k} \sum_{j=1}^k f(y_j) \pm O\left(\frac{1}{\sqrt{n}}\right)$ is an asymptotically valid confidence interval. It is instructive to compare this with the standard confidence interval construction $\frac{1}{k} \sum_{j=1}^k f(y_j) \pm O\left(\frac{1}{\sqrt{k}}\right)$ that would be used if one treated the $y_j$'s as if they were i.i.d.\ observations, e.g., if a practitioner were given the dataset and did not realize the observations were synthetically generated. Thus, we see that if $k = o(n)$, the last confidence interval construction would still be asymptotically valid.

Extending the preceding example in a different direction, one might wish to perform a $\chi^2$ test based on private synthetic data. The usual $\chi^2$ statistic takes the form $n\sum_{\ell=1}^K \frac{(p_\ell - q_\ell)^2}{q_\ell}$, where the $q_\ell$'s are the expected proportions and the $p_\ell$'s are the observed proportions in each of $K$ categories. With synthetic data, the $p_\ell$'s, being binomial proportions of the form $\frac{1}{n} \sum_{i=1}^n f(x_i)$, would be approximated using the synthetic data proportions $\frac{1}{k} \sum_{j=1}^k f(y_j)$. In order to obtain an asymptotically valid $\chi^2$-test, we could bound the error of the two statistics in terms of the accuracy bound on the private proportions, and add a suitable correction.

\paragraph{\textbf{Covariance matrix estimation.}} Now suppose we have observations $x_i \in \real^p$, forming a data matrix $X \in \real^{n \times p}$. Also suppose the distribution of the $x_i$'s has bounded support. By assigning a function $f \in \mathcal{F}$ for each of the $p^2$ matrix entries, we can control the accuracy of the sample covariance matrix $\frac{X^TX}{n}$ based on a private synthetic dataset $Y \in \real^{m \times k}$. By standard results \cite{DwoEtal14a}, it follows that we can bound the spectral norm error $\left\|\frac{X^TX}{n} - \frac{Y^TY}{m}\right\|_2$. Matrix perturbation techniques \cite{SteSun90} then allow us to translate such an accuracy bound into deviation bounds on the eigenvectors/eigenvalues of the private covariance matrix relative to the original dataset. Thus, we can run a downstream machine learning task such as principal component analysis, and be guaranteed that the principal components are an accurate approximation to those computed from the non-private dataset.

Another immediate application is to graphical model estimation. Recall that for a multivariate Gaussian distribution, the edges of the corresponding graphical model correspond to nonzero entries of the inverse covariance matrix. As described in the preceding paragraph, the covariance matrix can be estimated accurately (in spectral norm) using a private synthetic dataset. By standard results \cite{BloEtal12, KarEtal11}, this means the inverse covariance matrix can also be estimated accurately, so the edge structure of a graphical model estimated from the private synthetic data would still closely estimate the edge structure of a graphical model estimated from the original dataset. Note that the assumption that the data are Gaussian would somewhat contradict the assumption for bounded data imposed in our theory, but we could either apply the theorems with a bound which grows logarithmically with the data matrix dimensions, or else rely on results which show that the inverse covariance matrix can still be used to estimate the edge structure of graphical models in certain (limited) discrete data settings~\cite{LohWai13}.

\paragraph{\textbf{$M$-estimation.}} First consider the setting of ordinary least squares linear regression, where the goal is to estimate an underlying parameter vector $\beta^* \in \real^p$ based on observations $(X, y)$, where $X \in \real^{n \times p}$ and $y \in \real^p$. When $X^TX$ is invertible, the least squares estimator takes the form $\hat{\beta}_{OLS} = \left(\frac{X^TX}{n}\right)^{-1} \frac{X^Ty}{n}$. As described in the preceding example, our theory allows us to control the accuracy of the covariance matrix $\frac{X^TX}{n}$ in private synthetic data. Similarly (assuming the domain of the $y_i$'s is bounded as well), we can bound the accuracy of $\frac{X^Ty}{n}$, leading to an overall bound on the deviation between the ordinary least squares estimator computed on the private synthetic dataset and the same estimator computed on the non-private dataset.

The assumption that the distribution of the $y_i$'s has bounded support naturally raises the question of whether guarantees could also be derived for methods such as logistic regression. The non-existence of a closed-form solution means we have to rely on techniques for general empirical risk minimization. Note that while we can choose an objective function which is bounded for different values of $\beta$, we need to guarantee that the objective function is accurate for all possible $\beta$ in the set over which we perform minimization. Thus, our theory could be applied, for instance, over a grid of values. If the grid were fine enough, Lipschitz properties of the loss function would imply that the accuracy of the objective function is uniformly bounded, say, on a compact set. Fortunately, many machine learning algorithms exhibit Lipschitz continuity~\cite{vonBou04, BubSel21, Kov22}. Thus, our method is applicable to a wide range of machine learning tasks. Furthermore, it is well-known that Lipschitz properties are beneficial from the point of view of robustness~\cite{HeiAnd17, CohEtal19}. The same techniques could be used, e.g., to obtain guarantees on a Lasso estimator computed on private data in the setting of high-dimensional linear regression, where one cannot obtain a closed-form solution for the estimator as in the case of ordinary least squares.


\subsection{Computational considerations}
\label{SecComp}

It is known that no polynomial-time algorithm exists which can generate privatize synthetic data and preserve all two-dimensional marginals up to accuracy $o(1)$ \cite{UllEtal11}. To avoid this, Boedihardjo et al.~\cite{BoeEtal22} assumed the data were drawn i.i.d.\ from a distribution, leading to an average-case rather than worst-case analysis. Note that Algorithm \ref{AlgBoeEtal} can be implemented as a linear program with $|\mathcal{F}|+m+1$ constraints and has runtime polynomial in $p$. Also, note that Algorithm~\ref{AlgWasZho} is polynomial in $n$, and Algorithm~\ref{AlgAldRub} is polynomial in both $n$ and $\frac{1}{\epsilon}$. In their follow-up paper, Boedihardjo et al.~\cite{BoeEtal22b} again worked on generating synthetic data by comparing marginals. They enhanced the upper bound on overall error-per-query to $O(\sqrt n p^4$) by minimizing a covariance loss instead of the difference in linear statistics. However, the latter algorithm exhibits exponential complexity in $p$. Their procedures \cite{BoeEtal22, BoeEtal22b} involve the addition of noise, which may introduce systematic error. In the paper \cite{BoeEtal22c}, they proposed a noiseless approach that also relies on preserving marginals; however, a limitation of this approach is that the runtime of the algorithm is stochastic.


\section{Conclusion}

In this paper, we have supplied methods for generating private synthetic data for both discrete and continuous distributions. Our theoretical gurantees hold under certain conditions on the underlying distributions (e.g., lower and upper bounds on the density over the support of the distribution, a Lipschitz property, or boundedness of the domain), and it would certainly be of interest to further study how these conditions might be relaxed. In particular, the case where the data are generated from a continuous distribution with support $[0,B]^p$ could be studied in a straightforward manner using the framework of analysis laid out in our paper. However, it would be useful to be able to derive results valid for unbounded domains, both in the discrete and continuous settings, as well.

Another worthwhile avenue to explore would be to relax the condition that the set $\mathcal{F}$ of functions is finite. As discussed in Section~\ref{SecApplications}, this is important for applications such as empirical risk minimization. Some general theoretical analysis has appeared recently~\cite{AsaLoh23} involving the combinatorial dimension of the function space; it would be natural to try to combine such results with our algorithmic framework and study the overall guarantees in specific applications.

A third direction would be to provide synthetic data generation mechanisms for other notions of privacy. For instance, if one wishes to satisfy the more general notion of $(\epsilon, \delta)$-differential privacy, one could add Gaussian rather than Laplace noise in the optimization step of Algorithm~\ref{AlgBoeEtal}. Several papers, including the follow-up work by Boedihardjo et al.~\cite{BoeEtal22c}, provide alternatives for private synthetic data generation that do not involve the addition of any noise, as do the papers~\cite{LuEtal17,AbaEtal19}. However, these noiseless approaches often lack rigorous accuracy~\cite{LuEtal17,AbaEtal19} or privacy~\cite{LuEtal17} guarantees.

We also mention the important statistical question regarding optimality of our approaches, and a rigorous characterization of tradeoffs between privacy and accuracy. Our theoretical results impose certain lower bounds on $(n,m,k)$, opening the door for a much more careful analysis where one tries to determine whether it is possible to devise alternative strategies for private synthetic data generation with smaller sample/data set sizes, say, if the goal is to achieve a private data set with a certain bound on the level of privacy and accuracy. In terms of optimality with regards to the synthetic dataset size $k$, this bears some resemblance to the notion of ``thinning" a dataset for optimal compression~\cite{DwiMac21}.

Lastly, our paper has focused entirely on theoretical aspects, but it would be interesting to see how our method performs in simulations. A concrete comparison would be to take the uniform-sampling version of Algorithm~\ref{AlgBoeEtal} proposed by Boedihardjo et al.~\cite{BoeEtal22} in the case of Boolean data and compare it with our version, which first privately estimates a discrete histogram.
It would be interesting to see whether the computational speedups are tangible, and how different private synthetic data generation algorithms perform with regard to accuracy over different function families $\mathcal{F}$.


\section*{Acknowledgments}

This work was supported by the Cantab Capital Institute for the Mathematics of Information via the Philippa Fawcett Internship programme (Faculty of Mathematics, University of Cambridge). The authors thank the anonymous reviewers for their feedback, which improved the clarity of the manuscript.


\newpage

\bibliographystyle{IEEEtran}
\bibliography{refs}

\newpage

\appendices


\section{Auxiliary results}
\label{AppAux}

\begin{lemma}\label{l2}
Let $X_1, X_2, \ldots, X_n$ be i.i.d.\ Lap$(\sigma)$ random variables, with probability density function $p_\sigma(x) \propto \exp\left(-\frac{|x|}{\sigma}\right)$, and define
\[X_{\max} = \max \{X_1, X_2, \ldots, X_n\}.\]
With probability at least $1-\gamma$, we have $X_{\max} < \frac{\sigma \log n}{\gamma}$.
\end{lemma}

\begin{proof}
Let $Z_i = \max\{X_1, \dots, X_i\}$ for each $i$. By Markov's inequality, we have
\begin{equation*}
\mprob\left(X_{\max} \ge \frac{\sigma \log n}{\gamma}\right) \le \frac{\E[Z_n]}{\frac{\sigma \log n}{\gamma}},
\end{equation*}
so it suffices to show that $\E[Z_n] \le \sigma \log n$.

Denoting $E_i = \E[Z_i]$, note that for $t > 0$, we can calculate the cdf
\begin{equation*}
F_{Z_i}(t) = \left(1-\frac{1}{2} \exp\left(-\frac{t}{\sigma}\right)\right)^i, \quad \text{for all } i,
\end{equation*}
so
\begin{align*}
E_i & = \int_0^\infty \left(1-F_{Z_i}(t)\right) dt \\
& = \int_{0}^{\infty} 1 - \left(1 - \frac{1}{2}e^{-t/\sigma}\right)^i dt.
\end{align*}
It follows that
\begin{align*}
E_i - E_{i-1} & = \int_{0}^{\infty} \frac{e^{-t/\sigma}}{2} \left(1 - \frac{e^{-t/\sigma}}{2}\right)^{i-1} dx\\
& = \frac{\sigma}{i} \left(1 - \left(\frac{1}{2}\right)^i \right).
\end{align*}
Since $E_0 = 0$, we can recursively obtain
\begin{align*}
E_n & = \sum_{i=1}^{n}\frac{\sigma}{i} \left(1 - \left(\frac{1}{2}\right)^i \right) \\
& \leq \sum_{i=1}^{n}\frac{\sigma}{i} \leq \sigma \int_{1}^{n}\frac{1}{x}\,dx = \sigma \log n.
\end{align*}
%
%
\end{proof}



\begin{lemma}\label{lema1}
For two probability distributions $p$ and $q$ supported on a finite set $A$, we have
$$\min\limits_{x\in A} \left(\frac{p(x)}{q(x)}\right) \geq \frac{1}{1+{TV(p, q)}/{\min\limits_{x \in A}p(x)}}.$$
\end{lemma}

\begin{proof}
Suppose $\min\limits_{x\in A} \left(\frac{p(x)}{q(x)}\right)=\frac{p(y)}{q(y)}$, where $y \in A$. Denote $\delta = TV(p, q)$. Since $\delta=\sup\limits_{\mathcal{X} \subseteq A} |p(\mathcal{X}) - q(\mathcal{X})|$, we have $\delta \geq |p(y)-q(y)|$. This implies that
\begin{align*}
\min_{x \in A} \frac{p(x)}{q(x)} & = \frac{p(y)}{q(y)} \ge \frac{p(y)}{p(y) + \delta} = \frac{1}{1+\frac{\delta}{p(y)}} \\
& \ge \frac{1}{1+{\delta}/{\min\limits_{x \in A}p(x)}}.
\end{align*} 
\end{proof}

For the following lemmas, recall that the R\'enyi divergence of order $\alpha$ is defined by
\begin{equation}
\label{EqnRenyiGeneral}
D_{\alpha}(P\|Q) = \frac{1}{\alpha -1} \log \left( \sum_{i=1}^{n} \frac{p_i^{\alpha}}{q_i^{\alpha -1}} \right).
\end{equation}

\begin{lemma}
\label{LemMironov}
[Corollary 4 of Mironov \cite{Mir17}]
Let $P$, $Q$, and $R$ be distributions with the same support. For $\alpha > 1$, we have
\[
D_\alpha(P\|Q) \leq \frac{\alpha - 1/2}{\alpha - 1} D_{2\alpha}(P\|R) + D_{2\alpha - 1}(R\|Q).
\]
\end{lemma}

\begin{lemma}
\label{LemSasVer}
[Theorem 3 of Sason and Verd\'u~\cite{SasVer15}]
Suppose two distributions $p$ and $q$ are supported on a finite set $A$. Let $\beta_1 = \min\limits_{x \in A} \frac{q(x)}{p(x)}$. For $\alpha \in [0,1) \cup (1,\infty)$, we have
\[
D_{\alpha}(P || Q) \leq \frac{1}{\alpha-1} \log \left( 1 + TV(p,q) \cdot \frac{(\beta_1^{1-\alpha}-1)}{1-\beta_1} \right).
\]
\end{lemma}

\begin{theorem}
\label{app:2}
[Theorem 2 of Berend and Kontorovich~\cite{BerKon12}]
Let $X$ be a random variable supported on a discrete set of cardinality $k$, with density $p = (p_1, \dots, p_k)$, and let $X_1, X_2, \ldots, X_n$ be i.i.d.\ copies of $X$. Define the canonical estimator
\[\hat{p}_j^{(n)} = \frac{1}{n} \sum_{i=1}^{n} 1\{X_i=j\}, \quad \text{for all } 1 \le j \le k.\]
For every
$t \geq \sqrt{\frac{k}{n}}$, we have
\[
\mprob\left(\lvert \lvert \hat{p}_j^{(n)} - p_j\lvert \lvert_1 > t\right) \leq \exp\left(-\frac{n}{2} \left(t - \frac{k}{n}\right)^2\right).
\]
\end{theorem}


\section{Uniform sampling}
\label{AppUniform}

In this section, we apply Algorithm~\ref{AlgBoeEtal} to a dataset whose elements are sampled from $\{0,1,\dots,t-1\}^p$. The main challenge is to find an upper bound for the R\'enyi condition number.

%

\begin{theorem} \label{ThmBoeEtalUni}
Let $\epsilon > 0$, $\delta \in (0, \frac{1}{2}]$, and $\gamma\in (0, \frac{1}{4})$. 
Suppose that the true data $X=(X_1, X_2, \ldots, X_n)$ are sampled independently from $\Omega=\{0,1,\dots,t-1\}^p$ according to an unknown probability measure $\nu$. Set $\sigma = \delta/\log\left(|\mathcal{F}|/\gamma\right)$.
Suppose
the probability measure $\mu$ is uniform over $\Omega$. Suppose
\begin{align*}
    &n \geq \frac{2|\mathcal{F}|}{\epsilon \delta} \log\left(\frac{|\mathcal{F}|}{\gamma}\right), \qquad \min\{n, k\} \geq \frac{1}{\delta^2} \log\left(\frac{|\mathcal{F}|}{\gamma}\right), \\
    & m \geq \frac{|\Omega||\mathcal{F}|}{\delta^2 \gamma}.
\end{align*}
Then Algorithm~\ref{AlgBoeEtal} outputs an $\epsilon$-differentially private dataset $Y$ which, with probability at least $1-4\gamma$, is $8\delta$-accurate.
\end{theorem}

\begin{proof}
First note that since $\mu(x) = \frac{1}{|\Omega|}$ for all $x$, we have
\begin{equation*}
\kappa(\nu \| \mu) = \sum_{x \in \Omega} \frac{\nu^2(x)}{\mu(x)} = |\Omega| \sum_{x \in \Omega} \nu^2(x) \le |\Omega| \sum_{x \in \Omega} \nu(x) = |\Omega|.
\end{equation*}
The statement then follows by applying Theorem~\ref{ThmBoeEtal}.
\end{proof}



\section{Discrete case}

\subsection{Sampling from a perturbed histogram}
\label{app:4}

We review an algorithm from Wasserman and Zhou~\cite{WasZho10}, used to privately estimate the density $\mu$ to be used as an input in Algorithm~\ref{AlgBoeEtal}. Note that Algorithm~\ref{AlgWasZho} was originally designed for estimating a continuous density on $[0,1]^p$; here, we adapt it to a discrete setting, where the number of bins is determined by the support of our distribution, rather than an appropriate discretization of the unit cube. Accordingly, for our purposes, we take the number of bins $m = |\Omega| = t^p$.

\begin{algorithm}
    \caption{Perturbed histogram~\cite{WasZho10}}\label{AlgWasZho}
    \begin{algorithmic}
        \STATE \textbf{Input:} Dataset $X=\left(x_1, \ldots, x_n\right)$, privacy parameter $\epsilon$
                \STATE 1. Partition $X$ into $m$ bins $B_1, B_2, \dots, B_m$, such that \\ each bin is a cube with sides of length $h=\frac{1}{m^p}$
                \STATE 2. Let $C_j$ denote the number of true data points falling within the bin $B_j$, and denote the probability measure before perturbation by $\nu_1(j)=C_j/n$
                \STATE 3. Add noise: Define $D_j=\max\{C_j+\omega_j, 0\}$,\\ where $\omega_1, \omega_2, \dots, \omega_m$ are i.i.d.\ draws from a Laplace \\ density $g(\omega)=\frac{\epsilon}{4}e^{-|\omega|\frac{\epsilon}{2}}$
                \STATE \textbf{Output:} The histogram estimate $Z=(\mu_1, \mu_2, \dots, \mu_m)$, where $\mu_j=D_j/\sum \limits_{s}{D_s}$
    \end{algorithmic}
\end{algorithm}




\subsection{Proof of Theorem~\ref{t1}}
\label{AppT1}

Let $\nu_1$ be the empirical distribution of $X$ before perturbation. Motivated by the proof of Theorem 4.4 in Wasserman and Zhou~\cite{WasZho10}, we first prove that
\begin{equation}
\label{EqnUniform}
\max \limits_x |\nu_1(x)-\mu(x)| \le \frac{2\sigma (1+\tau_2 |\Omega|) \log |\Omega|}{\gamma n},
\end{equation}
with high probability. The notation used is the same as in Algorithm \ref{AlgWasZho}.

Note that
\begin{equation}
\label{alg2 1}
\mu_j=\frac{(C_j+\omega_j)_+}{\sum \limits_s (C_s+\omega_s)_+} =\frac{(C_j+\omega_j)_+}{n}\frac{n}{\sum \limits_s (C_s+\omega_s)_+}.
\end{equation}
Let $M = \max\{|\omega_1|, |\omega_2|, \ldots, |\omega_m|\}$, and note that by Lemma~\ref{l2}, we have $M \le \frac{\sigma}{\gamma} \log |\Omega|$, with probability at least $1-\gamma$. On this event, we have
\begin{align*}
\nu_1(j) -\frac{|\omega_j|}{n} & \leq \nu_1(j) +\frac{\omega_j}{n}=\frac{C_j+\omega_j}{n} \\
& \leq \frac{(C_j+\omega_j)_+}{n}\leq \nu_1(j) +\frac{|\omega_1|}{n},
\end{align*}
so we obtain
\begin{equation}
\label{alg2 2}
\left|\frac{(C_j+\omega_j)_+}{n}-\nu_1(j)\right| \leq \frac{|\omega_j|}{n}\leq \frac{M}{n} \le \frac{\sigma \log |\Omega|}{\gamma n}.
\end{equation}
Similarly, we have
\[1-\frac{\sum_{s}|\omega_s|}{n} \leq 1+\frac{\sum_{s}\omega_s}{n} \leq \frac{\sum_s(C_s+\omega_s)_+}{n} \leq 1+\frac{\sum_{s}|\omega_s|}{n},\]
so we obtain
\begin{align}
\label{alg2 3}
    \left|\frac{\sum_s(C_s+\omega_s)_+}{n}-1\right| & \leq \frac{\sum_{s}|\omega_s|}{n} \le \frac{M|\Omega|}{n} \notag \\
& \le \frac{\sigma |\Omega| \log |\Omega|}{\gamma n}.
\end{align}
Note that under the assumptions, we have $\frac{\sigma |\Omega| \log |\Omega|}{\gamma n} \le \frac{1}{2}$, so $\frac{\sum_s (C_s + \omega_s)_+}{n} \ge \frac{1}{2}$.

If we define
\begin{align*}
\varepsilon_1 & := \frac{(C_j + \omega_j)_+}{n} - \nu_1(j), \\
\varepsilon_2 & := \frac{n}{\sum_s (C_s + \omega_s)_+} - 1,
\end{align*}
we see that
\begin{align}
\label{EqnPrelim}
|\mu_j - \nu_1(j)| & = \left|(\nu_1(j) + \varepsilon_1)(1+\varepsilon_2) - \nu_1(j)\right| \notag \\
& = |\nu_1(j) \varepsilon_2 + \varepsilon_1 + \varepsilon_1 \varepsilon_2|,
\end{align}
and inequalities~\eqref{alg2 2} and~\eqref{alg2 3} imply that
\begin{align*}
|\varepsilon_1| & \le \frac{\sigma \log |\Omega|}{\gamma n}, \\
|\varepsilon_2| & \le \frac{n}{\sum_s (C_s + \omega_s)_+} \left|\frac{\sum_s (C_s + \omega_s)_+}{n} - 1\right| \\
& \le \frac{2\sigma |\Omega| \log |\Omega|}{\gamma n}.
\end{align*}

Note that by a Chernoff bound~\cite{Ver18}, we have $\nu_1(j) \le \tau_2$ for each $j$, with probability at least $1-\exp(-c_0 n)$. Taking a union bound and noting that $n \gg \log |\Omega|$ by assumption, we see that $\sup_j \nu_1(j) \le \tau_2$ with probability at least $1 - \exp(-cn)$. Applying the triangle inequality to equation~\eqref{EqnPrelim}, we then have
\begin{align*}
|\mu_j - \nu_1(j)| & \le 2|\varepsilon_1| + \tau_2 |\varepsilon_2| \\
& \le \frac{2\sigma(1+\tau_2 |\Omega|) \log |\Omega|}{\gamma n},
\end{align*}
which is the bound~\eqref{EqnUniform}.

Now define $\delta_1= TV(\nu,\nu_1)$, $\delta_2=TV(\mu,\nu_1)$,  $\beta_1=\min\limits_{x} \left(\frac{\nu_1(x)}{\nu(x)}\right)$, and $\beta_2=\min\limits_x \left(\frac{\mu(x)}{\nu_1(x)}\right)$. 
Recall the formula~\eqref{EqnRenyiGeneral}, and note that we are interested in the case $\alpha = 2$. From Lemmas~\ref{LemMironov} and~\ref{LemSasVer}, provided in Appendix~\ref{AppAux}, we have
\begin{align}
\label{eq1}
& D_2(\nu || \mu) \leq \frac{3}{2}D_4(\nu || \nu_1) + D_3(\nu_1 || \mu) \notag \\
&  \leq \frac{1}{2} \log \left(1+\delta_1\frac{\beta_1^{-3}-1}{1-\beta_1}\right) + \frac{1}{2} \log \left(1+\delta_2\frac{\beta_2^{-3}-1}{1-\beta_2}\right).
\end{align}


From Lemma \ref{lema1}, we have
\begin{align*}
\beta_1 & \geq \frac{1}{1+\delta_1/\min\limits_{x \in \Omega}\nu_1(x)}, \\
\beta_2 & \geq \frac{1}{1+\delta_2/\min\limits_{x \in \Omega}\mu(x)}.
\end{align*}

\subsubsection{Upper bound on $\delta_1$}

By Theorem~\ref{app:2}, we have
\begin{equation*}
    P(\lVert \nu_1 - \nu \rVert_1 > a) \leq e^{-\frac{n}{2}(a-\sqrt{|\Omega|/n})^2}, \text{ for } a \geq \sqrt{|\Omega|/n}.
\end{equation*}
Setting $a=\sqrt{\frac{2\log \frac{1}{\gamma}}{n}}+\sqrt{\frac{|\Omega|}{n}}$, we obtain 
\begin{equation}\label{eq2}
    \delta_1 = \frac{1}{2}||\nu_1 - \nu||_1  \leq \frac{1}{2}\left(\sqrt{\frac{2\log \frac{1}{\gamma}}{n}}+\sqrt{\frac{|\Omega|}{n}}\right),
\end{equation}
with probability at least $1-\gamma$. \\

\subsubsection{Lower bound on $\min_x \nu_1(x)$}

By the triangle inequality, for every $ x\in X $, we have
\begin{equation*}\label{inqnu1}
\nu_1(x) \geq \nu(x) - |\nu_1(x) - \nu(x)| \geq \nu(x) - \|\nu-\nu_1\|_1.
\end{equation*}
From inequality~\eqref{eq2}, we deduce that
\begin{align*}
\min \limits_{x} \nu_1(x) &\geq \tau_1 - \sqrt{\frac{2\log \frac{1}{\gamma}}{n}} - \sqrt{\frac{|\Omega|}{n}}
\geq \frac{\tau_1}{2},
\end{align*}
provided
\begin{equation*}
n \geq \frac{8}{\tau_1^2}\left(|\Omega| + 2\log \frac{1}{\gamma}\right).
\end{equation*}
This also implies
\begin{equation*}
\beta_1 \ge \frac{1}{1 + \frac{\tau_1/4}{\tau_1/2}} = \frac{2}{3}.
\end{equation*}


\subsubsection{Upper bound on $\delta_2$}

Using the bound~\eqref{EqnUniform}, we have
\begin{align*}
\delta_2 & = \frac{1}{2}\lVert \nu_1 - \mu \rVert_1 \leq \frac{\sigma |\Omega|(1+\tau_2 |\Omega|) \log |\Omega|}{\gamma n} \\
& = \frac{\delta |\Omega| (1 + \tau_2 |\Omega|) \log |\Omega|}{\gamma n \log (|\mathcal{F}|/\delta)}.
\end{align*}


\subsubsection{Upper bound on $\min_x \mu(x)$}

By the triangle inequality, for every $x$, we obtain
\begin{equation*}
\label{inqmu}
\mu(x) \geq \nu_1 (x) - |\mu(x) - \nu_1(x)| \geq \nu_1 (x) - \|\nu_1 - \mu\|_1.
\end{equation*}
Thus, for
\begin{equation*}
n \geq \frac{8\delta |\Omega|(1+\tau_2|\Omega|) \log |\Omega|}{\tau_1 \gamma \log(|\mathcal{F}|/\delta)},
\end{equation*}
we have
\begin{equation*}
\min_x \mu(x) \ge \frac{\tau_1}{2} -  \frac{2\delta |\Omega|(1+\tau_2 |\Omega|) \log |\Omega|}{\gamma n \log (|\mathcal{F}|/\delta)} \ge \frac{\tau_1}{4}.
\end{equation*}
%
%
%
%
This also implies
\begin{equation*}
\beta_2 \ge \frac{1}{1 + \frac{\tau_1/8}{\tau_1/4}} = \frac{2}{3}.
\end{equation*}


\subsubsection{Upper bound on $K= e^{D_2(\nu\parallel\mu)} $}

Now we incorporate our bounds:
\begin{equation*}\label{d}
\begin{split}
K = e^{D_2(\nu||\mu)}
\leq \left(1+C_1 \left(\sqrt{\frac{2\log \frac{1}{\gamma}}{n}}+\sqrt{\frac{|\Omega|}{n}}\right)\right)^{\frac{1}{2}} \\
 \qquad \cdot \left(1+\frac{C_2\delta |\Omega| (1+\tau_2 |\Omega|) \log |\Omega|}{\gamma n \log (|\mathcal{F}|/\delta)}\right)^{\frac{1}{2}},
\end{split}
\end{equation*}
with probability at least $1-3\gamma - \exp(-cn)$, by a union bound.
\subsection{Performance evaluation}
In this section, we outline the performance of our method in simulations and present a comparison between the uniform-sampling version of Algorithm~\ref{AlgBoeEtal} proposed by Boedihardjo et al.~\cite{BoeEtal22} for Boolean data and our version, which initially privately estimates a discrete histogram. Here is the setup:
\begin{itemize}
\item  $X$ is uniformly sampled from $\Omega=\{0,1\}^p$ 
\item Family of test functions $\mathcal{F}$ encodes all $d$-dimensional marginals (as described in the Section~\ref{SecBackground}) 
\item  Parameters: $d = 2$, $\epsilon = 0.2$, $\delta = 0.5$, $\gamma = 0.25$
\item The sizes of $n$, $k$, and $m$ are equal to their respective lower bounds as stated in Theorem~\ref{ThmBoeEtalUni} (for uniform sampling) and Corollary~\ref{cor} (for perturbed histogram).
\item The Python library `pulp' is utilized to solve optimization problems (the simplex method is used)
\end{itemize}
Note that computational speedups are tangible, and both methods perform similarly in terms of accuracy (Fig.~\ref{fig:execution_time} and~\ref{fig:error}).
 \begin{figure}[htbp]
    \centering
    \includegraphics[width=0.41\textwidth]{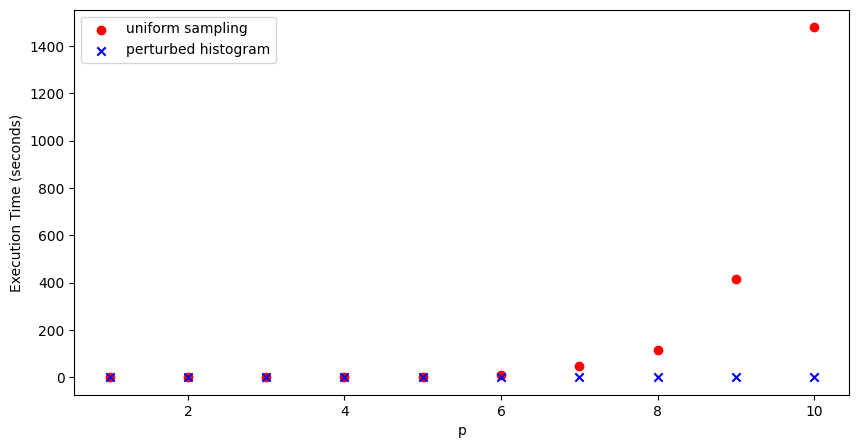}
    \caption{\footnotesize Comparison of Execution Time}
    \label{fig:execution_time}
\end{figure}

\begin{figure}
    \centering
    \includegraphics[width=0.41\textwidth]{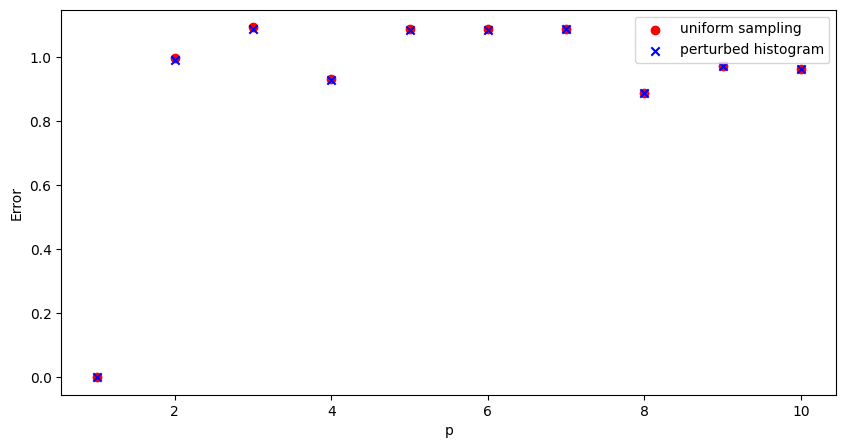}
    \caption{\footnotesize Comparison of Error}
    \label{fig:error}
\end{figure}
\section{Continuous case}
\label{appcon}
\subsection{The Bernstein mechanism}
\label{alg 3}
The Bernstein mechanism (Algorithm~\ref{AlgAldRub}) offers a solution to privately approximate a function $F_D$ by utilizing iterated Bernstein basis polynomials and perturbing the coefficients with noise. Recall that the \emph{sensitivity} of a function $f: \Omega^n \rightarrow \mathbb{R}^d$ is given by
$$S(f)=\sup_{ X \sim X'} \| f(X)- f(X') \|_1,$$
where the supremum is taken over all $X, X' \in \mathcal{X}^n$ that differ in one entry~\cite{DwoEtal06}. The sensitivity of a function $F: \Omega^n \times \mathcal{Y} \rightarrow \mathbb{R}^d$ is defined as $S(F)=\sup\limits_{y \in \mathcal{Y}} S(F(\cdot, y))$. 

\begin{definition}
 Let $k$ be a positive integer. The Bernstein basis polynomials of degree $k$ are defined as $b_{\ell, k}(y)=\binom{k}{\ell} y^\ell(1-$ $y)^{k-\ell}$, for $\ell=0, \ldots, k$.
 \end{definition}
 
 \begin{definition}
      Let $f:[0,1] \rightarrow \mathbb{R}$ and let $k$ be a positive integer. The Bernstein polynomial of $f$ of degree $k$ is defined as $B_k(f ; y)=\sum_{\ell=0}^k f(\ell / k) b_{\ell, k}(y)$.
      
 If $h$ is a positive integer, the iterated Bernstein operator of order $h$ is defined as the sequence of linear operators
 $$B_k^{(h)}=I-\left(I-B_k\right)^h=\sum_{i=1}^h\binom{h}{i}(-1)^{i-1} B_k^i,$$
 where $I=B_k^0$ denotes the identity operator and $B_k^i$ is defined inductively as $B_k^i=B_k \circ B_k^{i-1}$ for $i \geq 1$.
\end{definition}

It is known~\cite{Mic73} that the iterated Bernstein polynomial of order $h$ can be computed as
$$
B_k^{(h)}(f ; y)=\sum_{\ell=0}^k f\left(\frac{\ell}{k}\right) b_{\ell, k}^{(h)}(y),
$$
where $b_{\ell, k}^{(h)}(y)=\sum_{i=1}^h\binom{h}{i}(-1)^{i-1} B_k^{i-1}\left(b_{\ell, k} ; y\right)$.

\begin{definition}
 Assume $f:[0,1]^{p} \rightarrow \mathbb{R}$ and let $k_1, \ldots, k_{p}$, and $h$ be positive integers. The  iterated Bernstein polynomial of $f$ of order $h$ is defined as
$$
B_{k_1, \ldots, k_p}^{(h)}(f ; y)=\sum_{j=1}^p \sum_{\ell_j=0}^{k_j} f\left(\frac{\ell_1}{k_1}, \cdots, \frac{\ell_p}{k_p}\right) \prod_{i=1}^p b_{\ell_i, k_i}^{(h)}\left(y_i\right).
$$
\end{definition}

\begin{algorithm}
    \caption{Bernstein mechanism~\cite{AldRub17}}\label{AlgAldRub}
    \begin{algorithmic}
        \STATE \textbf{Input:} Dataset $X \in \left([0,1]^p\right)^n$, target function $F: \left([0,1]^p\right)^n \times [0,1]^p \rightarrow \real$ with sensitivity $S(F)$, \\  cover size $k$, Bernstein order $j$, privacy budget $\epsilon > 0$
                \STATE 1. Define the lattice cover $P = \left\{0, \frac{1}{k}, \frac{2}{k}, \ldots, 1\right\}^p$ \\ of $[0,1]^p$
                \STATE 2. Calculate the perturbation scale $\lambda = \frac{S(F)(k + 1)}{\epsilon}$
                \STATE 3. For each $p \in P$, perturb $F(X,p)$\\ by adding i.i.d.\ Laplace noise: $\widetilde{F}(X,p) \leftarrow F(X,p) + Z$, where $Z \sim \text{Lap}(\lambda)$
                \STATE \textbf{Output:} Function estimate $\sum_{j=1}^p \sum_{\ell_j=0}^k \widetilde{F}\left(X, \frac{\ell_1}{k}, \cdots, \frac{\ell_p}{k}\right) \prod_{i=1}^p b_{\ell_i, k}^{(j)}\left(y_i\right)$    \end{algorithmic}
\end{algorithm}


\subsection{Kernel Density Estimation}
\label{app ker}

We will need the following result about the convergence of kernel density estimators.

\begin{theorem} [Theorem 1 of Jiang~\cite{Jia17}] \label{ThmJiang}
Suppose $X_1, \dots, X_n$ are i.i.d.\ samples from a distribution $\nu$. Define
\begin{equation*}
\begin{split}
f_1(r) = \nu(x) - \inf\limits_{x' \in B(x,r)}\nu(x'), \\
f_2(r) = \sup\limits_{x' \in B(x,r)}{\nu(x')} - \nu(x).
\end{split}
\end{equation*}
Let $H_0$ be a symmetric, positive definite matrix such that $\|H_0\|_2 = 1$, and define $H=h^2H_0$, where $h \geq \big(\log n/n\big)^{\frac{1}{p}}$. With probability at least $1-\frac{1}{n}$, we have
    \begin{equation} \label{eqker1}
       \nu_1(x) \leq \nu(x) + \alpha + C\sqrt{\frac{\log n}{n h^p}}+64p^2k(0)\frac{\log n}{n h^p}
    \end{equation}
if $\int_{\mathbb{R}^p} K(u) f_2(h|u|/\sigma_p(H_0)) \, du \leq \alpha,$ and
\begin{equation}\label{eqker2}
    \nu_1(x) \geq \nu(x) - \alpha - C\sqrt{\frac{\log n}{n h^p}}-64p^2k(0)\frac{\log n}{n h^p}
\end{equation}
if $\int_{\mathbb{R}^p} K(u)f_1(h|u|/\sigma_p(H_0)) \, du  \leq \alpha $, where $C = 8p\sqrt{v_d \|\nu\|_\infty}\int_{0}^{\infty} k(t) t^{\frac{p}{2}}\,dt$, and $\sigma_1(H_0) \geq \ldots \geq \sigma_p(H_0) > 0$ denote the eigenvalues of $H_0$.
\end{theorem}


\subsection{Proof of Theorem~\ref{t2}}
\label{AppT2}


The kernel density estimator $\nu_1$ inherits all continuity and differentiability properties of the kernel $K$, making it a smooth function with derivatives of all orders. Furthermore, the sensitivity can easily be bounded (cf.\ \textit{Examples, Kernel Density Estimation} in Section 6 of Alda and Rubinstein~\cite{AldRub17}). This leads to the following result:

\begin{theorem}[Theorem 3 of Ald\`a and Rubinstein~\cite{AldRub17}] \label{thD2}
Let $\gamma \in (0, 1/4]$.
With probability at least $1-\gamma$, we have
\[
\sup_x |\nu_1(x) - \mu(x)| \leq O\left(\left( \frac{1}{h n \varepsilon \sqrt{(2\pi)^p}} \log \frac{1}{\gamma} \right)^{\frac{j}{p+j}}\right),
\]
where $\mu$ is the release of $\nu_1$ obtained via Algorithm~\ref{AlgAldRub}. This mechanism ensures $\epsilon$-differential privacy.
\end{theorem}

As in the proof of Theorem~\ref{t1}, we use the notation $\delta_1=$ TV ($\nu,\nu_1$), $\delta_2=$ TV ($\mu,\nu_1$),  $\beta_1=\min\limits_{x} \left(\frac{\nu_1(x)}{\nu(x)}\right)$, and $\beta_2=\min\limits_x \left(\frac{\mu(x)}{\nu_1(x)}\right)$. We will also use inequality~\eqref{eq1} to bound the R\'{e}nyi divergence.

\subsubsection{Bounds on $\min_{x} \nu_1(x)$ and $\delta_1$}

Applying Theorem~\ref{ThmJiang} (noting that $k(0)=1/(2\pi)^{p/2}$), we have
\begin{equation*} \label{inker1}
        \sup \limits_{x}|\nu(x)-\nu_1(x)| \leq \alpha + C_1 \sqrt{\frac{\log n}{n h^p}}+64p^2\frac{\log n}{n h^p \sqrt{(2\pi)^p}},
\end{equation*}
with probability at least $1 - \frac{1}{n}$.
Let $\alpha=\frac{\tau}{9}$ and $h=C^{'}\Big(\frac{\log n}{n}\Big)^\frac{1}{p}$, such that  $\frac{C_1}{\sqrt{C^{'}}} < \frac{\tau}{9}$ and $\frac{64p^2}{\big({C^{'}\sqrt{(2\pi)}}\big)^p}< \frac{\tau}{9}$. Plugging into Theorem~\ref{thD2}, we obtain the following inequalities with probability $1-\frac{1}{n}$ (apply the triangle inequality as shown in Appendix \ref{inqnu1}):
\begin{equation*}
  \delta_1=    \sup \limits_x |\nu(x)-\nu_1(x)| \leq \frac{\tau}{3} \implies \min \limits_{x} \nu_1(x) \geq \tau - \frac{\tau}{3} = \frac{2}{3}\tau.
\end{equation*}
\begin{claim}
We have
\[\sup \limits_x|\nu(x)-\nu_1(x)| \leq \frac{\tau}{3},\]
    with probability at least $1-\frac{1}{n}$, when $\frac{\log n}{n}\geq \max \left\{\big(\frac{\tau}{6C_1}\big)^2h^p,\frac{\tau h^p \sqrt{(2\pi)^p}}{384p^2}\right\}$.
\end{claim}
\begin{proof}
Since $\nu$ is $L$-Lipschitz, we have $f_1(h|u|) \leq Lh|u|$. This implies that
\begin{align*}
    \int_{\mathbb{R}^p} K(u)f_1(h|u|) \, du & \leq Lh\int_{\mathbb{R}^p} K(u)|u|\, du = Lh\mathbb{E}|u|.
\end{align*}
Similarly, we have
\begin{align*}
    \int_{\mathbb{R}^p} K(u)f_2(h|u|) \, du & \leq Lh\int_{\mathbb{R}^p} K(u)|u|\, du = Lh\mathbb{E}|u|.
\end{align*}
The rest of the proof consists of applying Theorem~\ref{thD2} with $\alpha=Lh\mathbb{E}|u|$ and $h=C^{'}\Big(\frac{\log n}{n}\Big)^\frac{1}{p}$.
\end{proof}
\subsubsection{Bounds on $\delta_2$ and $\min_x \mu(x)$}

By Theorem~\ref{thD2}, we have
\[\delta_2 \leq C\left( \frac{2}{h n \varepsilon \sqrt{(2\pi)^p}} \log \frac{1}{\gamma} \right)^{\frac{1}{p+1}} \leq \frac{\tau}{3}\] for $n \geq \Big(\frac{3C}{\tau} \Big)^{p+1}\Big(\frac{2}{h \varepsilon \sqrt{(2\pi)^p}} \log \frac{1}{\gamma}\Big) $, with probability at least $1-\gamma$. Therefore, it is true with probability $1-\gamma-\frac{1}{n}$ that (apply the triangle inequality as shown in Appendix \ref{inqmu})
\[\min \limits_{x} \mu(x) \geq \tau - \frac{\tau}{3} - \frac{\tau}{3} = \frac{\tau}{3}.\]

\subsubsection{Lower bounds on $\beta_1$ and $\beta_2$}

Analogously to the proof of Theorem~\ref{t1}, we obtain
\begin{equation*}
\beta_1 \ge \frac{1}{1+\frac{\tau/3}{2\tau/3}} = \frac{2}{3}
\end{equation*}
and
\begin{equation*}
\beta_2 \ge \frac{1}{1+\frac{\tau/3}{\tau/3}} = \frac{1}{2}.
\end{equation*}

\subsubsection{Upper bound on $K$}

Refer to Appendix~\ref{AppT1} (equation~\eqref{eq1}) for the previously established upper bound on $D(\nu || \mu)$. By incorporating this with our bounds, we can obtain
\[ K = \left(1 + C_1'\tau\right)^{\frac{1}{2}} \left(1 + C_2'\tau\right)^{\frac{1}{2}}, \]
with probability at least $1-\frac{1}{n}-\gamma$.

The rest of proof comes from plugging our results into Theorem~\ref{ThmBoeEtal}.

\end{document}